\documentclass{llncs}
\usepackage{makeidx}  
 \usepackage{amsmath}
\usepackage{enumitem}
\usepackage{amssymb}
\usepackage{bussproofs}
\usepackage{cite}

 \newtheorem{theoremA}{Theorem}
 \newtheorem{lemmaA}{Lemma}
 \newtheorem{definitionA}{Definition}
 \newtheorem{exampleA}{Example}

\begin{document}
\frontmatter          
\mainmatter              
\title{Automated Proof Search System For Logic Of Correlated Knowledge}
\titlerunning{Automated Proof Search System For Logic Of Correlated Knowledge}  
%

\author{Haroldas Giedra, Romas Alonderis}
\institute{Institute of Data Science and Digital Technologies, Vilnius University, Akademijos str. 4, LT-08663 Vilnius, Lithuania\\
\email{haroldas.giedra@mii.vu.lt, romas.alonderis@mii.vu.lt}}

\maketitle              

\begin{abstract}
The automated proof search system and decidability for logic of correlated knowledge is presented in this paper. The core of the proof system is the sequent calculus with the properties of soundness, completeness, admissibility of cut and structural rules, and invertibility of all rules. The proof search procedure based on the sequent calculus performs automated terminating proof search and allows us to achieve decision result for logic of correlated knowledge.

\keywords{Logic of correlated knowledge, sequent calculus, automated proof system, decidability, soundness, completeness, admissibility of the cut rule.}
\end{abstract}
\vspace{2mm}

\section*{Introduction}

Logic of correlated knowledge (LCK) has been introduced by Alexandru Baltag and Sonja Smets in \cite{BaltagSmets2010}. LCK is an epistemic logic enriched by observational capabilities of agents. Applications of the epistemic logic cover fields such as distributed systems, merging of knowledge bases, robotics or network security in computer science and artificial intelligence. By adding observational capabilities to agents, logic of correlated knowledge can be applied to reason about systems where knowledge correlate between spatially distributed parts of the system. This includes any social system, distributed information system, traffic light system, quantum system or any other system where knowledge is correlated.\\\\
Complex system may consist of one or more elementary parts. Associating agent to each part, we get multi-agent system, where agents can perform observations and get results. Allowing communication between agents, correlations between knowledge in distributed parts can be extracted. This can not be done by traditional epistemic logic or logic of distributed knowledge.\\\\
Our main scientific result is proof search system GS-LCK-PROC for logic of correlated knowledge, which lets to reason about knowledge automatically. The core of the system is the sequent calculus GS-LCK with the properties of soundness, completeness, admissibility of cut and structural rules, and invertibility of all rules. The ideas of semantic internalization, suggested by Sara Negri in \cite{Negri2005}, are used to get such properties for the calculus. The calculus provides convenient means for backward proof search and decision procedure for logic of correlated knowledge. The procedure generates a finite model for each sequent. As a result we get termination of the proof search and decidability of logic of correlated knowledge.\\\\
We start by defining syntax, semantics, and the Hilbert style proof system for logic of correlated knowledge in section \ref{sec:LogicOfCorrelatedKnowledge}. In section \ref{sec:GentzenStyleSequentCalculusGSLCK} we present Gentzen style sequent calculus GS-LCK and the properties of the proof system. Soundness and completeness of the GS-LCK and the properties of admissibility of weakening, contraction and cut are proved in sections \ref{sec:ProofOfSoundnessOfGSLCK},  \ref{sec:ProofOfThePropertiesOfGSLCK} and \ref{sec:ProofOfCompletenessOfGSLCK}. Automated proof search system GS-LCK-PROC and decidability of logic of correlated knowledge are presented in the final section \ref{sec:DecidabilityOfLogicOfCorrelatedKnowledge}.

\section{Logic of correlated knowledge}
\label{sec:LogicOfCorrelatedKnowledge}
\subsection{Syntax}
Consider a set \(N=\{a_{1},a_{2},...,a_{n}\}\) of agents. Each agent can perform its local observations. Given sets \(O_{a_{1}},..., O_{a_{n}}\) of possible observations for each agent, a joint observation is a tuple of observations \(o = (o_{a})_{a \in N}\in O_{a_{1}}\times ... \times O_{a_{n}}\) or \(o= (o_{a})_{a \in I}\in  O_{I}\), where \(O_{I}:= \times _{a \in I} O_{a}\) and \(I \subseteq N\). Joint observations together with results \(r \in R\) make new atomic formulas \(o^{r}\). 
\\\\
Each agent can know some information, and it is written as \(K_{a_{1}} A\) or \(K_{\{a_{1}\}} A\), which means that the agent \(a_{1}\) knows \(A\). A group of agents can also know some information and this is denoted by \(K_{\{a_{1},a_{2},a_{3}\}} A\) or \(K_{I} A\), where \(I=\{a_{1},a_{2},a_{3}\}\). A more detailed description about the knowledge operator \(K\) is given in \cite{Fagin1992,Hoek1997}.
\\\\
Syntax of logic of correlated knowledge is defined as follows:
\begin{definitionA}[Syntax of logic of correlated knowledge]\label{definitionA:1}
The language of logic of correlated knowledge has the following syntax:
\begin{displaymath}
A:=p\mid o^{r} \mid \neg A \mid A\vee B \mid A\wedge B \mid A \rightarrow B \mid K_{I}A
\end{displaymath}
where p is any atomic proposition, \(o = (o_{a})_{a \in I}\in O_{I}, r \in R\), and \(I \subseteq N\).
\end{definitionA} 
\subsection{Semantics}
Consider a system, composed of \(N\) components or locations. Agents can be associated to locations, where they will perform observations. States (configurations) of the system are functions \(s:O_{a_{1}}\times ... \times O_{a_{n}}\rightarrow R\) or \(s_{I}:O_{I}\rightarrow R\), where \(I \subseteq N\) and a set of results \(R\) is in the structure \((R,\Sigma)\) together with an abstract operation \(\Sigma: \mathcal{P}(R)\rightarrow R\) of composing results. The operation \(\Sigma\) maybe partial (defined only for some subsets \(A \subseteq R\)), but it is required to satisfy the condition: \(\Sigma\{\Sigma A_{k}: k \in K \} = \Sigma (\cup_{k \in K} A_{k})\) whenever \(\{A_{k}: k \in K \}\) are pairwise disjoint. \(\mathcal{P}(R)\) is a power set of \(R\).  For every joint observation \(e \in O_{I}\), the local state \(s_{I}\) is defined as:
\begin{displaymath}
s_{I}((e_{a})_{a \in I}) := \Sigma \{s(o):o\in O_{a_{1}}\times ... \times O_{a_{n}} \:\text{such that}\: o_{a}=e_{a} \:\text{for all}\: a\in I\}
\end{displaymath}  
If \(s\) and \(t\) are two possible states of the system and a group of agents \(I\) can make exactly the same observations in these two states, then these states are observationally equivalent to \(I\), and it is written as \(s\overset{I}{\sim}t\). Observational equivalence is defined as follows: 
\begin{definitionA}[Observational equivalence]\label{definitionA:1}
Two states \(s\) and \(t\) are observationally equivalent \(s\overset{I}{\sim}t \:\:\text{iff}\:\:s_{I}=t_{I}\).
\end{definitionA} 
A model of logic of correlated knowledge is a multi-modal Kripke model \cite{Kripke1963}, where the relations between states  mean observational equivalence. It is defined as:
\begin{definitionA}[Model of logic of correlated knowledge]\label{definitionA:1}
For a set of states \(S\), a family of binary relations \(\{\overset{I}{\sim}\}_{I \subseteq N} \subseteq S\times S\) and a function of interpretations \(V:S\rightarrow (P \rightarrow \{true, false\})\), where \(P\) is a set of atomic propositions, a model of logic of correlated knowledge is a multi-modal Kripke model \((S, \{\overset{I}{\sim}\}_{I \subseteq N}, V)\) that satisfies the following conditions:
\begin{enumerate}
\item 
For each \(I \subseteq N\), \(\overset{I}{\sim}\) is a multi-modal equivalence relation;
\item 
Information is monotonic: if \(I \subseteq J\), then \(\overset{J}{\sim} \subseteq \overset{I}{\sim} \);
\item 
Observability principle: if \(s\overset{N}{\sim}s'\), then \(s=s'\);
\item 
Vacuous information: \(s\overset{\emptyset}{\sim}s'\) for all \(s,s'\in S\).
\end{enumerate} 
\end{definitionA} 
The satisfaction relation \(\models\)  for model \(M\), state \(s\) and formulas \(o^{r}\) and \(K_{I}A\) is defined as follows:
\begin{itemize}[label=$\bullet$]
\item \(M,s\models K_{I}A\) iff \(M,t\models A\) for all states \(t\overset{I}{\sim}s\).
\item \(M,s\models o^{r}\) \hspace{8px} iff \(s_{I}(o)=r\).
\end{itemize} 
The formula \(K_{I}A\) means that the group of agents \emph{I} carries the information that \(A\) is the case, and \(o^{r}\) means that \(r\) is the result of the joint observation \(o\).
\\\\
If formula \(A\) is true in any state of any model, then it is named as a valid formula.

\subsection{Hilbert style calculus HS-LCK}
Alexandru Baltag and Sonja Smets defined the Hilbert style calculus for logic of correlated knowledge in \cite{BaltagSmets2010}. Fixing a finite set \(N=\{a_{1},...,a_{n}\}\) of agents, a finite result structure \((R,\Sigma)\) and a tuple of finite sets \(\vec O = (O_{a_{1}},...,O_{a_{n}})\) of observations, for every set \(I, J\subseteq N\), every joint observation \( o\in O_{I}, \:O_{I}= \times _{a \in I} O_{a}\), and results \( r, p\in R \), the Hilbert style calculus for logic of correlated knowledge over \((R,\Sigma,\vec O)\) is as follows:
\begin{itemize}[label=$\bullet$]

\item Axioms:

\begin{description}

\item[H1.] \(A\rightarrow (B \rightarrow A ) \)
\vspace{1mm}

\item[H2.] \((A\rightarrow (B \rightarrow C )) \rightarrow ((A \rightarrow B ) \rightarrow (A \rightarrow C))  \)
\vspace{1mm}

\item[H3.] \( (\neg A\rightarrow \neg B ) \rightarrow (B \rightarrow A ) \)
\vspace{1mm}

\item[H4.] \(K_{I}(A\rightarrow B)\rightarrow (K_{I} A \rightarrow K_{I} B) \) \hspace*{5mm}   \textit{(Kripke's axiom)}
\vspace{1mm}

\item[H5.] \(K_{I}A\rightarrow A\) \hspace*{36mm}
\textit{(Truthfulness)}
\vspace{1mm}

\item[H6.] \(K_{I}A\rightarrow K_{I}K_{I}A\) \hspace*{27mm}
\textit{(Positive introspection)}
\vspace{1mm}

\item[H7.] \(\neg K_{I}A\rightarrow K_{I}\neg K_{I}A\) \hspace*{22mm}
\textit{(Negative introspection)}
\vspace{1mm}

\item[H8.] \(K_{I}A\rightarrow K_{J}A,  \:\:\text{where}\:\: I\subseteq J\) \hspace*{9mm}
\textit{(Monotonicity of group knowledge)}
\vspace{1mm}

\item[H9.] \(A\rightarrow K_{N}A\) \hspace*{35mm}
\textit{(Observability)}
\vspace{1mm}

\item[H10.] \( \underset{o\in O_{I}}{ \wedge} \:\underset{r\in R}{ \vee} o^{r} \)   \hspace*{32mm} \textit{(Observations always yield results)}
\vspace{1mm}

\item[H11.] \(o^{r} \rightarrow \neg o^{p}, \:\:\text{where}\:\: r \neq p \)  \hspace*{14mm} \textit{(Observations have unique results)}
\vspace{1mm}
\vspace{1mm}

\item[H12.] \(o^{r}_{I} \rightarrow K_{I} o^{r}_{I} \)  \: \hspace*{31mm} \textit{(Groups know the results of}
\\\hspace*{57mm} \textit{their joint observations)}
\vspace{1mm}
\vspace{1mm}

\item[H13.] \((\underset{o\in O_{I}}{\wedge} o^{r_{o}} \wedge K_{I}A)\rightarrow  K_{\emptyset}(\underset{o\in O_{I}}{\wedge} o^{r_{o}} \rightarrow A)\)  
\\\hspace*{55mm}  \textit{(Group knowledge is correlated }
\\\hspace*{57mm} \textit{knowledge (i.e. is based on joint}
\\\hspace*{57mm} \textit{observations))}
\vspace{1mm}
\vspace{1mm}

\item[H14.] \(\underset{o\in \bar e}{\wedge} o^{r_{o}} \rightarrow e^{\Sigma \{r_{o}:o \in \bar e\}}, \:\:\text{where}\:\:\:\:\:\:\:\:\:\: e\in O_{I}, \: \bar e := \{o= (o_{i})_{i \in N}\in O_{i_{1}}\times ... \times O_{i_{n}}: \hspace*{4mm} o_{i}=e_{i} \:\text{for all}\: i\in I\}\).
\hspace*{18mm} \textit{(Result composition axiom)}
\vspace{1mm}
\vspace{1mm}

\end{description}

\item Rules:
\begin{displaymath}
\frac{A,A\rightarrow B}{B} \:\:(Modus\:ponens) \:\:\:\:\:\:\:\:\:\:\:\:\:\:\:\:\:\:\:\:
\frac{A}{K_{I}A} \:\:(K_{I} - necessitation) \:\:\:\:\:\:\:\:\:\:\:\:\:\:\:\:\:\:\:\:
\end{displaymath}

\end{itemize}
Sets \(I, J\) may be empty in axioms H4 - H8 and in rule \((K_{I} - necessitation)\).\\\\
The Hilbert style calculus HS-LCK for logic of correlated knowledge is sound and complete with respect to correlation models over \((R,\Sigma,\vec O)\) \cite{BaltagSmets2010}.

\section{Gentzen style sequent calculus GS-LCK}
\label{sec:GentzenStyleSequentCalculusGSLCK}
Gerhard Gentzen introduced sequent calculus in 1934 \cite{Gentzen1934}. Sequents in the system GS-LCK are statements of the form \(\Gamma\Rightarrow \Delta\), where \(\Gamma\) and \( \Delta\) are finite, possibly empty multisets of relational atoms \(s\overset{I}{\sim}t\) and labelled formulas \(s:A\), where \(s, t \in S\), \(I\subseteq N\) and \(A\) is any formula in the language of logic of correlated knowledge. The formula \(s:A\) means \(s\models A\), and \(s\overset{I}{\sim}t\) is an observational equivalence or relation between the states in the model of logic of correlated knowledge.\\\\
The sequent calculus consists of axioms and rules. Applying rules to the sequents, a proof-search tree for the root sequent is constructed. If axioms are in all the leaves of the proof-search tree, then the root sequent is called as a provable sequent and \(\Delta\) follows  from \(\Gamma\) of the root sequent.\\\\
Fixing a finite set \(N=\{a_{1},...,a_{n}\}\) of agents, a finite result structure \((R,\Sigma)\) and a tuple of finite sets \(\vec O = (O_{a_{1}},...,O_{a_{n}})\) of observations, for every set \(I, J\subseteq N\), every joint observation \( o\in O_{I}, \:O_{I}= \times _{a \in I} O_{a}\), and results \( r, p\in R \), the Gentzen style sequent calculus GS-LCK for logic of correlated knowledge over \((R,\Sigma,\vec O)\) is as follows:
\begin{itemize}[label=$\bullet$]
\item Axioms:
\begin{itemize}[label=$\centerdot$]
\item \(s:p,\Gamma\Rightarrow \Delta, s:p.\)
\item \(s:o^{r},\Gamma\Rightarrow \Delta, s:o^{r}.\)
\item \(s:o^{r_{1}},s:o^{r_{2}},\Gamma\Rightarrow \Delta\), where \(r_{1}\neq r_{2}.\)
\end{itemize}
\vspace{1mm}
\vspace{1mm}

\item Propositional rules:
\begin{displaymath}
\frac{\Gamma \Rightarrow \Delta, s:A}{s:\neg A,\Gamma\Rightarrow\Delta}\:\:(\neg\Rightarrow) \:\:\:\:\:\:\:\:\:\:\:\:\:\:\:\:\:\:\:\:\:\:\:\:\:\:\:\:\:\:\:\:\:\:\:\:\:\:\:\:\:\:\:\:\:\:\:\:\:\:
\frac{s: A,\Gamma \Rightarrow\Delta}{\Gamma\Rightarrow \Delta, s:\neg A}\:\:(\Rightarrow\neg) \:\:\:\:\:\:\:\:\:\:\:\:\:\:\:\:\:\:
\end{displaymath}
\vspace{1 pt}
\begin{displaymath}
\frac{ s:A, \Gamma\Rightarrow \Delta\:\:\:\:\:\:\:\:s:B, \Gamma\Rightarrow \Delta}{ s:A\vee B, \Gamma\Rightarrow \Delta}\:\:(\vee\Rightarrow) \:\:\:\:\:\:\:\:\:\:\:\:\:\:\:\:\:\:\:\:\:\:\:\:\:\:\:\:\:\:\:\:\:\:\:\:\:\:\:\:\:\:\:\:\:\:
\frac{\Gamma \Rightarrow\Delta,s:A,s:B}{\Gamma\Rightarrow\Delta, s:A\vee B}\:\:(\Rightarrow\vee) \:\:\:\:\:\:\:\:\:\:\:\:\:\:\:\:\:\:\:\:\:\:\:\:\:\:\:\:\:\:\:\:\:\:\:\:\:\:\:\:\:\:\:\:\:\:\:\:\:\:
\end{displaymath}
\vspace{1 pt}
\begin{displaymath}
\frac{s:A,s:B,\Gamma \Rightarrow\Delta}{s:A\wedge B,\Gamma\Rightarrow\Delta}\:\:(\wedge\Rightarrow) \:\:\:\:\:\:\:\:\:\:\:\:\:\:\:\:\:\:\:\:\:\:\:\:\:\:\:\:\:\:\:\:\:\:\:\:\:\:\:\:\:\:\:\:\:\:\:\:\:\:
\frac{\Gamma\Rightarrow \Delta, s:A\:\:\:\:\:\:\:\:\Gamma\Rightarrow \Delta, s:B}{\Gamma\Rightarrow \Delta, s:A\wedge B}\:\:(\Rightarrow\wedge) \:\:\:\:\:\:\:\:\:\:\:\:\:\:\:\:\:\:
\end{displaymath}
\vspace{1 pt}
\begin{displaymath}
\frac{ \Gamma\Rightarrow \Delta,s:A\:\:\:\:\:\:\:\:s:B, \Gamma\Rightarrow \Delta}{ s:A\rightarrow B, \Gamma\Rightarrow \Delta}\:\:(\rightarrow\Rightarrow) \:\:\:\:\:\:\:\:\:\:\:\:\:\:\:\:\:\:\:\:\:\:\:\:\:\:\:\:\:\:\:\:\:\:\:\:\:\:\:\:\:\:\:\:\:\:
\frac{s:A,\Gamma \Rightarrow\Delta,s:B}{\Gamma\Rightarrow\Delta, s:A\rightarrow B}\:\:(\Rightarrow\rightarrow) \:\:\:\:\:\:\:\:\:\:\:\:\:\:\:\:\:\:\:\:\:\:\:\:\:\:\:\:\:\:\:\:\:\:\:\:\:\:\:\:\:\:\:\:\:\:\:\:\:\:
\end{displaymath}
\vspace{1 pt}
\vspace{1mm}

\item Knowledge rules:
\begin{displaymath}
\frac{t:A, s:K_{I}A, s\overset{I}{\sim}t, \Gamma\Rightarrow\Delta}{s:K_{I}A, s\overset{I}{\sim}t, \Gamma\Rightarrow\Delta}\:\:( \: K_{I}\Rightarrow) \:\:\:\:\:\:\:\:\:\:\:\:\:\:\:\:\:\:\:\:\:\:\:\:\:\:\:\:\:\:\:\:\:\:\:\:\:\:\:\:\:\:\:\:\:\:\:\:\:\:
\frac{s\overset{I}{\sim}t, \Gamma\Rightarrow \Delta, t:A}{\Gamma\Rightarrow \Delta, s:K_{I}A}\:\:(\Rightarrow K_{I}) \:\:\:\:\:\:\:\:\:\:\:\:\:\:\:\:\:\: 
\end{displaymath}
The rule \((K_{I}\Rightarrow)\) requires that \(I\neq N\) and \(t:A\) be not in \(\Gamma\). The rule \((\Rightarrow K_{I})\) requires that \(I\neq N\) and \textit{t} be not in the conclusion. Set \(I\) maybe an empty set in both rules.
\begin{displaymath}
\frac{s:A, s:K_{N}A, s\overset{N}{\sim}s, \Gamma\Rightarrow\Delta}{s:K_{N}A, s\overset{N}{\sim}s, \Gamma\Rightarrow\Delta}\:\:( \: K_{N}\Rightarrow) \:\:\:\:\:\:\:\:\:\:\:\:\:\:\:\:\:\:\:\:\:\:\:\:\:\:\:\:\:\:\:\:\:\:\:\:\:\:\:\:\:\:\:\:\:\:\:\:\:\:
\frac{s\overset{N}{\sim}s, \Gamma\Rightarrow \Delta, s:A}{\Gamma\Rightarrow \Delta, s:K_{N}A}\:\:(\Rightarrow\:K_{N}) 
\end{displaymath}
The rule \((K_{N}\Rightarrow)\) requires that \(s:A\) be not in \(\Gamma\). The rule \((\Rightarrow K_{N})\) requires that \(s:A\) be not in \(\Delta\).
\vspace{1 pt}
\vspace{1mm}

\item Observational rules:
\begin{displaymath}
\frac{s\overset{I}{\sim}t,\{s:o^{r_{o}} \}_{o\in O_{I}}, \{t:o^{r_{o}} \}_{o\in O_{I}}, \Gamma\Rightarrow  \Delta}{\{s:o^{r_{o}} \}_{o\in O_{I}}, \{t:o^{r_{o}} \}_{o\in O_{I}}, \Gamma\Rightarrow  \Delta}\:\:(OE) \:\:\:\:\:\:\:\:\:\:\:\:\:\:\:\:\:\:\:\:\:\:\:\:\:\:\:\:\:\:\:\:\:\:\:\:\:\:\:\:\:\:\:\:\:\:\:\:\:\:\:\:\:\:\:\:\:\:\:\:\:\:\:\:
\end{displaymath}
The rule \((OE)\) requires that \(I\neq \emptyset\) and formulas \(s\overset{I}{\sim}t\), \(s:o^{r_{o}}\) and \(t:o^{r_{o}}\) be not in \(\Gamma\), where \(o\in O_{I}\). 
\vspace{1 pt}
\begin{displaymath}
\frac{\{s:o^r_I,\Gamma \Rightarrow \Delta \}_{r\in R}}{\Gamma\Rightarrow\Delta} \:\:(OYR) \:\:\:\:\:\:\:\:\:\:\:\:\:\:\:\:\:\:\:\:\:\:\:\:\:\:\:\:\:\:\:\:\:\:\:\:\:\:\:\:\:\:\:\:\:\:\:\:\:\:\:\:\:\:\:\:\:\:\:\:\:\:\:\:\:\:\:\:\:\:\:\:\:\:\:\:\:\:\:\:\:\:\:\:\:\:\:\:\:\:\:\:\:\:\:\:\:\:\:\:\:\:\:\:\:\:\:\:\:\:\:\:\:\:\:\:\:\:\:\:\:\:\:\:\:\:\:\:\:\:\:\:
\end{displaymath}
The rule \((OYR)\) requires: 
\begin{enumerate}
\item \(s:o^{r}_I\) be not in \(\Gamma\) for all \(r \in R\) and \(s:o^{r_{1}}_I\) be in \(\Delta\) for some \(r_{1} \in R\).
\item \(I\neq \emptyset\).
\end{enumerate}
\vspace{12 pt}
\begin{displaymath}
\frac{s:e^{\Sigma \{r_{o_N}:o_N \in \bar e\}}_I, \{s:o^{r_{o_N}}_N\}_{o_N\in \bar e}, \Gamma \Rightarrow \Delta}{\{s:o^{r_{o_N}}_N\}_{o_N\in \bar e},\Gamma\Rightarrow\Delta}\:\:(CR) \:\:\:\:\:\:\:\:\:\:\:\:\:\:\:\:\:\:\:\:\:\:\:\:\:\:\:\:\:\:\:\:\:\:\:\:\:\:\:\:\:\:\:\:\:\:\:\:\:\:\:\:\:\:\:\:\:\:\:\:\:\:\:\:\:\:\:\:\:\:\:\:
\end{displaymath}
The rule \((CR)\) requires that \(s:e^{\Sigma \{r_{o_N}:o_N \in \bar e\}}_I\) be not in \(\Gamma\). 
\vspace{10 pt}
\item Substitution rules:
\begin{displaymath}
\frac{s:p, t:p, s\overset{N}{\sim}t, \Gamma\Rightarrow\Delta}{ t:p, s\overset{N}{\sim}t, \Gamma\Rightarrow\Delta}\:\:( \: Sub(p) \Rightarrow) \:\:\:\:\:\:\:\:\:\:\:\:\:\:\:\:\:\:\:\:\:\:\:\:\:\:\:\:\:\:\:\:\:\:\:\:\:\:\:\:\:\:\:\:\:\:\:\:\:\:
\frac{s:o^{r}, t:o^{r}, s\overset{I}{\sim}t, \Gamma\Rightarrow\Delta}{ t:o^{r}, s\overset{I}{\sim}t, \Gamma\Rightarrow\Delta}\:\:( \: Sub(o^{r}) \Rightarrow) 
\end{displaymath}
\vspace{1mm}
The rules \((Sub(p) \Rightarrow)\) and \((Sub(o^{r})\Rightarrow)\) require that \(s:p\) and \(s:o^{r}\) be not in \(\Gamma\), accordingly. 
\\
\vspace{1 pt}
\item Relational rules:
\begin{displaymath}
\frac{ s\overset{I}{\sim}s, \Gamma\Rightarrow\Delta}{\Gamma\Rightarrow\Delta}\:\:(Ref) \:\:\:\:\:\:\:\:\:\:\:\:\:\:\:\:\:\:\:\:\:\:\:\:\:\:\:\:\:\:\:\:\:\:\:\:\:\:\:\:\:\:\:\:\:\:\:\:\:\:
\frac{ s\overset{I}{\sim}t, s\overset{I}{\sim}s', s'\overset{I}{\sim}t, \Gamma\Rightarrow\Delta}{s\overset{I}{\sim}s', s'\overset{I}{\sim}t,\Gamma\Rightarrow\Delta}\:\:(Trans)
\end{displaymath}
The rule \((Ref)\) requires that \(s\) be in the conclusion and \(s\overset{I}{\sim}s\) be not in \(\Gamma\). The rule \((Trans)\) requires that \(s\overset{I}{\sim}t\) be not in \(\Gamma\).
\begin{displaymath}
\frac{ s'\overset{I}{\sim}t, s\overset{I}{\sim}s', s\overset{I}{\sim}t, \Gamma\Rightarrow\Delta}{s\overset{I}{\sim}s', s\overset{I}{\sim}t,\Gamma\Rightarrow\Delta}\:\:(Eucl) \:\:\:\:\:\:\:\:\:\:\:\:\:\:\:\:\:\:\:\:\:\:\:\:  
\frac{  s\overset{I}{\sim}t, s\overset{J}{\sim}t, \Gamma\Rightarrow\Delta}{s\overset{J}{\sim}t,\Gamma\Rightarrow\Delta}\:\:(Mon) \:\:\:\:\:\:\:\:\:\:\:\:\:\:\:\:\:\:\:\:\:\:\:\:\:\:\:
\end{displaymath}
The rule \((Mon)\) stands for monotonicity and requires that \(I\subseteq J\). Sets \(I, J\) may be empty. The rules \((Eucl)\) and \((Mon)\) require that \(s'\overset{I}{\sim}t\) and \(s\overset{I}{\sim}t\) be not in \(\Gamma\), accordingly.
\vspace{1 pt}
\end{itemize}
\vspace{1 pt}
The sequent calculus GS-LCK is sound and complete with respect to correlation models over \((R,\Sigma,\vec O)\). It also has the beautiful properties of rule invertibility and admissibility of the cut and structural rules. It is crucial in making the automated proof system in the present paper.
\begin{theoremA}[Properties of GS-LCK]\label{Theorem:Properties}
The sequent calculus GS-LCK has the following properties:
\begin{itemize}[itemsep=2pt,topsep=0pt,parsep=0pt,partopsep=0pt]
\item Invertibility of rules.
\item Admissibility of weakening.
\item Admissibility of contraction.
\item Admissibility of cut.
\item Termination.
\end{itemize}
\end{theoremA}
Proofs of soundness, completeness, and the properties of GS-LCK are given in the next sections.

\section{Proof of soundness of GS-LCK}
\label{sec:ProofOfSoundnessOfGSLCK}
\begin{definitionA}[Extended syntax]\label{definitionA:1}
Extended syntax of LCK is as follows:
\begin{displaymath}
A:=s:A_{1} \mid s\overset{I}{\sim}t \mid s:A_{1}\vee A \mid s:A_{1}\wedge A \mid s:A_{1} \rightarrow A 
\end{displaymath}
\begin{displaymath}
A_{1}:=p\mid o^{r} \mid \bot\mid \top \mid \neg A_{1} \mid A_{1}\vee A_{2} \mid A_{1}\wedge A_{2} \mid A_{1}\rightarrow A_{2} \mid K_{I}A_{1}
\end{displaymath}
where p is any atomic proposition, \(o \in O_{I}, I \subseteq N, r \in R\) and \(s, t \in S\).
\end{definitionA} 

\begin{definitionA}[Extended semantics]\label{definitionA:1}
If \(s, t, v\in \mathbf{S}\) and \(M \in \mathbf{M}\), then the truthfulness of the formula in the state \(v\) of the model \(M\) is defined as follows:    
\begin{itemize}[itemsep=2pt,topsep=0pt,parsep=0pt,partopsep=0pt]
\item \(v \models s:A\) \:\: iff \:\(s \models A\).
\item \(v \models s\overset{I}{\sim}t\) \: iff \:\(s\overset{I}{\sim}t \in \mathbf{R}\).
\end{itemize}
\end{definitionA} 
Commas "," in \(\Gamma\) of the sequent \(\Gamma \Rightarrow \Delta \) mean conjunction "\(\wedge\)", commas "," in \(\Delta\) - disjunction "\(\vee\)". The arrow "\(\Rightarrow\)" stands for implication "\(\rightarrow\)".
\begin{definitionA}[Formula of the sequent]\label{definitionA:1}
If Seq is a sequent \(\Gamma \Rightarrow \Delta \), then \emph{the formula of the sequent} \(F(Seq)\) is obtained by:
\begin{enumerate}[itemsep=2pt,topsep=0pt,parsep=0pt,partopsep=0pt]
\item[1)] putting \(\Gamma \) and \(\Delta \) in parentheses;
\item[2)] replacing empty \(\Gamma\) by \(s:\top\);
\item[3)] replacing empty \(\Delta\) by \(s:\bot\);
\item[4)] replacing commas "," by conjunction "\(\wedge\)" in \(\Gamma\);
\item[5)] replacing commas "," by disjunction "\(\vee\)" in \(\Delta\);
\item[6)] replacing "\(\Rightarrow\)" by implication "\(\rightarrow\)".
\end{enumerate}
\end{definitionA} 
\begin{exampleA} \( F(Seq) := ( t:A_{1} \wedge s:K_{I}A_{1}\wedge s\overset{I}{\sim}t \wedge  t:A_{2} ) \rightarrow ( s:B_{1} \vee t:B_{2} )\) is the formula of the sequent \(Seq := t:A_{1}, s:K_{I}A_{1}, s\overset{I}{\sim}t,  t:A_{2} \Rightarrow s:B_{1}, t:B_{2}\).
\end{exampleA} 
\begin{definitionA}[Sequent without labels and relational atoms]\label{definitionA:1} 
If Seq is a sequent, then the sequent without labels and relational atoms of Seq is obtained removing all labels near formulas and all relational atoms from Seq.
\end{definitionA} 
\begin{lemmaA}[Validity of the formula of the sequent]\label{lemmaA:1}
If the formula of the sequent Seq is valid, then the formula of the sequent Seq without labels and relational atoms is valid, as well.
\end{lemmaA}
\begin{proof}\hspace*{1mm}\\
Suppose we have a set of states \(S\) of a model \(M\). For each formula of the sequent we have a tuple of its labels \((s_{1},...,s_{l})\in S \times ... \times S\). If the formula with labels \((s_{1},...,s_{l})\) is valid, then it is valid with substituted labels \((s',...,s')\), because \(\{(s',...,s'):s' \in S\} \subseteq \{(s_{1},...,s_{l}): s_{1},...,s_{l} \in S \}\). Having \(s \models s':A\), iff \(s' \models A\), we can remove the label \(s'\). \\
All relational atoms become \(s'\overset{I}{\sim}s', I \subseteq N\). They are valid because of reflexivity in models. Applying the rules of GS-LCK they appear only in the first argument of implication of the formula of the sequent. We can remove relational atoms, because having a valid formula \(( A_{1} \wedge ... \wedge  A_{l} ) \rightarrow ( B_{1} \vee ... \vee B_{k} ) \) and removing valid formula \(A_{i}\) from the first argument of implication, the validity is maintained.
\end{proof}

\begin{theoremA}[Soundness of GS-LCK]\label{theoremA:1}
If sequent S is provable in GS-LCK, then the formula of the sequent S without labels and relational atoms is valid with respect to correlation models over \((R,\Sigma,\vec O)\).
\end{theoremA}
\begin{proof}\hspace*{1mm}\\
We prove the validity of all axioms and soundness of all the rules of GS-LCK:
\begin{itemize}[label=$\bullet$]
\item Axioms:
\begin{itemize}[label=$-$]
\item Formula of the axiom \(s:p,\Gamma\Rightarrow s:p, \Delta\) is valid, because it is true in any state of any model. The same is for the axiom \(s:o^{r},\Gamma\Rightarrow s:o^{r}, \Delta\).
\item Validity of the formula of the axiom \(s:o^{r_{1}},s:o^{r_{2}},\Gamma\Rightarrow \Delta\), where \(r_{1}\neq r_{2}\), follows from the axiom "H11. \(o^{r} \rightarrow \neg o^{p}, \:\:\text{where}\:\: r \neq p \)".
\end{itemize}
\item Propositional rules as in \cite{Negri2001book}.
\item Knowledge rules:
\begin{itemize}[label=$-$]
\item Rule \((K_{I}\Rightarrow)\):
\begin{displaymath}
\frac{t:A, s:K_{I}A, s\overset{I}{\sim}t, \Gamma\Rightarrow\Delta}{s:K_{I}A, s\overset{I}{\sim}t,\Gamma\Rightarrow\Delta}\:\:( \: K_{I}\Rightarrow), \:\:\:\:\:  I\neq N. \:\:\:\:\:\:\:\:\:\:\:\:\:\:\:\:\:\:\:\:\:\:\:\:\:\:\:\:\:\:\:\:\:\:\:\:\:\:\:\:\:\:\:\:\:\:\:\:\:\:\:\:\:\:\:\:\:\:\:\:\:\:\:\:\:\:\:\:\:\:\:\:\:\:\:\:\:\:\:\:\:\:\:\:\:\:\:\:\:\:\:\:\:\:\:\:\:\:\:\:\:\:\:\:\:\:\:\:\:\:\:\:
\end{displaymath}
We prove by contraposition that, if the formula of the premise \((t:A, s:K_{I}A,\) \(s\overset{I}{\sim}t, \Gamma\Rightarrow\Delta)\) of the rule \((K_{I}\Rightarrow)\) is valid, then the formula of the conclusion \((s:K_{I}A, s\overset{I}{\sim}t,\Gamma\Rightarrow\Delta)\) is valid, too.
\\[1px]
The formula of the conclusion \((s:K_{I}A, s\overset{I}{\sim}t,\Gamma\Rightarrow\Delta)\) is false, when \(s:K_{I}A\), \(s\overset{I}{\sim}t\) and all formulas in \(\Gamma\) are true, and all formulas in \(\Delta\) are false. By semantic definition of the knowledge operator \(K_{I}\), formula \(A\) is true in all the states accessible from the state \(s\) by relation \(I\). States \(t\) are accessible from the state \(s\), because \(s\overset{I}{\sim}t\) is true, therefore the formula \(t:A\) is true. If \(t:A\), \(s:K_{I}A\), \(s\overset{I}{\sim}t\) and all formulas in \(\Gamma\) are true and all formulas in \(\Delta\) are false, then the formula of the premise \((t:A, s:K_{I}A, s\overset{I}{\sim}t, \Gamma\Rightarrow\Delta)\) is false.
\\*
\item Rule \(( \Rightarrow K_{I})\):
\begin{displaymath}
\frac{s\overset{I}{\sim}t, \Gamma\Rightarrow \Delta, t:A}{\Gamma\Rightarrow \Delta, s:K_{I}A}\:\:(\Rightarrow K_{I}), \:\:\:\:\: I\neq N\: \textit{and } t \textit{ is not in the conclusion}.\:\:\:\:\:\:\:\:\:\:\:\:\:\:\:\:\:\:\:\:\:\:\:
\end{displaymath}
The formula of conclusion \((\Gamma\Rightarrow \Delta, s:K_{I}A)\) is false, when all formulas in \(\Gamma\) are true and all formulas in \(\Delta\) and \(s:K_{I}A\) are false. If the formula \(s:K_{I}A\) is false, then there exists a state \(t\) accessible from state \(s\) by relation \(I\), where \(A\) is false. If \(s\overset{I}{\sim}t\) and all formulas in \(\Gamma\) are true and all formulas in \(\Delta\) and \(t:A\) are false, then the formula of the premise \((s\overset{I}{\sim}t, \Gamma\Rightarrow \Delta, t:A)\) is false.
\\[1px]
The label \(t\) cannot be in the conclusion, because we can get situations, where the formula of the premise \((s\overset{I}{\sim}t, \Gamma\Rightarrow \Delta, t:A)\) is valid and the formula of the conclusion \((\Gamma\Rightarrow \Delta, s:K_{I}A)\) is not. An example:
\begin{displaymath}
\frac{s\overset{I}{\sim}t, t:A \Rightarrow t:A}{t:A \Rightarrow s:K_{I}A}  \:\:(\Rightarrow K_{I}) \:\:\:\:\:\:\:\:\:\:\:\:\:\:\:\:\:\:\:\:\:\:\:\:\:\:\:\:\:\:\:\:\:\:\:\:\:\:\:\:\:\:\:\:\:\:\:\:\:\:\:\:\:\:\:\:\:\:\:\:\:\:\:\:\:\:\:\:\:\:\:\:\:\:\:\:\:\:\:\:\:\:\:\:\:\:\:\:\:\:\:\:\:
\end{displaymath}
\item The validity of the rules \((K_{N}\Rightarrow)\) and \(( \Rightarrow K_{N})\) is proved in the same way.
\end{itemize}
\item Observational rules:
\begin{itemize}[label=$-$]
\item Rule \((OYR)\):
\begin{displaymath}
\frac{\{s:o^r,\Gamma \Rightarrow \Delta \}_{r\in R}}{\Gamma\Rightarrow\Delta} \:\:(OYR) \:\:\:\:\:\:\:\:\:\:\:\:\:\:\:\:\:\:\:\:\:\:\:\:\:\:\:\:\:\:\:\:\:\:\:\:\:\:\:\:\:\:\:\:\:\:\:\:\:\:\:\:\:\:\:\:\:\:\:\:\:\:\:\:\:\:\:\:\:\:\:\:\:\:\:\:\:\:\:\:\:\:\:\:\:\:\:\:\:\:\:\:\:\:\:\:\:\:\:\:
\end{displaymath}
If \(R\) is a set of results, and \(o\) is a joint observation, then there exists a result \(r \in R\) that \(o^{r}\) is true. If there exists \(r\) that \(o^{r}\) is true and all formulas in \(\Gamma\) are true and all formulas in \(\Delta\) are false, then one formula of premises \((\{s:o^r,\Gamma \Rightarrow \Delta \}_{r\in R})\) is false.
\item Rule \((CR)\):
\begin{displaymath}
\frac{s:e^{\Sigma \{r_{o}:o \in \bar e\}}, \{s:o^{r_{o}}\}_{o\in \bar e}, \Gamma \Rightarrow \Delta}{\{s:o^{r_{o}} \}_{o\in \bar e},\Gamma\Rightarrow\Delta}\:\:(CR) \:\:\:\:\:\:\:\:\:\:\:\:\:\:\:\:\:\:\:\:\:\:\:\:\:\:\:\:\:\:\:\:\:\:\:\:\:\:\:\:\:\:\:\:\:\:\:\:\:\:\:\:\:\:\:\:\:\:\:\:\:\:\:\:\:\:
\end{displaymath}
The contraposition is proved by the axiom "H14. \(\underset{o\in \bar e}{\wedge} o^{r_{o}} \rightarrow e^{\Sigma \{r_{o}:o \in \bar e\}}\)".
\item The soundness of rules \((OE)\), \(( \: Sub(p) \Rightarrow)\) and \(( \: Sub(o^{r}) \Rightarrow)\) is proved in the same way.
\end{itemize}
\item Relational rules:
\begin{itemize}[label=$-$]
\item Rule \((Mon)\):
\begin{displaymath}
\frac{  s\overset{I}{\sim}t, s\overset{J}{\sim}t, \Gamma\Rightarrow\Delta}{s\overset{J}{\sim}t,\Gamma\Rightarrow\Delta}\:\:(Mon) \:\:\:\:\:\:\:\:\:\:\:\:\:\:\:\:\:\:\:\:\:\:\:\:\:\:\:\:\:\:\:\:\:\:\:\:\:\:\:\:\:\:\:\:\:\:\:\:\:\:\:\:\:\:\:\:\:\:\:\:\:\:\:\:\:\:\:\:\:\:\:\:\:\:\:\:\:\:\:\:\:\:\:\:\:\:\:\:\:\:\:\:\:\:\:\:\:\:\:\:\:\:\:\:\:\:\:\:
\end{displaymath}
The contraposition follows from condition to models of LCK: 2. If \(I \subseteq J\) then \(\overset{J}{\sim} \subseteq \overset{I}{\sim} \).
\item The validity of rules \((Ref)\), \((Trans)\) and \((Eucl)\) is proved in the same way.
\end{itemize}
\end{itemize}
We have proved the validity of all axioms and soundness of all the rules of GS-LCK. The statement of the theorem follows from lemma \ref{lemmaA:1}.
\end{proof}

\section{Proof of the properties of GS-LCK}
\label{sec:ProofOfThePropertiesOfGSLCK}

\begin{lemmaA}[Admissibility of contraction with atomic formulas]\label{Lemma:AdmissibilityOfContractionAtomic}
\hspace*{1mm}\\
If a sequent \((\Pi_{atomic},\Pi_{atomic}, \Gamma \Rightarrow \Delta, \Lambda_{atomic}, \Lambda_{atomic})\) is provable in GS-LCK, then the sequent \((\Pi_{atomic},\Gamma \Rightarrow \Delta, \Lambda_{atomic})\) is also provable with the same bound of the height of the proof in GS-LCK. \(\Gamma, \Delta\) are any multisets of formulas. \(\Pi_{atomic}, \Lambda_{atomic}\) are any multisets of atomic formulas \(s:p, s:o^{r}, s\overset{I}{\sim}t\).
\end{lemmaA}
\begin{proof}\hspace*{1mm}\\
Lemma \ref{Lemma:AdmissibilityOfContractionAtomic} is proved by induction on the height \(<h>\) of the proof of the sequent \((\Pi_{atomic},\Pi_{atomic},\Gamma \Rightarrow \Delta, \Lambda_{atomic}, \Lambda_{atomic})\).

\(<h=1>\)

If the sequent \((\Pi_{atomic},\Pi_{atomic},\Gamma \Rightarrow \Delta, \Lambda_{atomic}, \Lambda_{atomic})\) is an axiom, then the sequent \((\Pi_{atomic},\Gamma \Rightarrow \Delta, \Lambda_{atomic})\) is an axiom too.

\(<h>1>\)
\begin{itemize}[itemsep=2pt,topsep=0pt,parsep=0pt,partopsep=0pt]
\item The rule \((K_{I}\Rightarrow)\) was applied in the last step of the proof of the sequent.\\
\begin{itemize}[label=$-$,itemsep=2pt,topsep=0pt,parsep=0pt,partopsep=0pt]
\item One or two formulas of the principal pair is in \(\Pi_{atomic}\). 
\begin{displaymath}
\frac{t:A, s:K_{I}A, s\overset{I}{\sim}t, s\overset{I}{\sim}t, \Pi'_{atomic}, \Pi'_{atomic}, \Gamma'\Rightarrow\Delta, \Lambda_{atomic}, \Lambda_{atomic}}{s:K_{I}A, s\overset{I}{\sim}t, s\overset{I}{\sim}t, \Pi'_{atomic}, \Pi'_{atomic}, \Gamma'\Rightarrow\Delta, \Lambda_{atomic}, \Lambda_{atomic}}\:\:( \: K_{I}\Rightarrow) \:\:\:\:\:\:\:\:\:\:\:\:\:\:\:\:\:\:\:\:\:\:\:\:\:\:\:\:\:\:\:\:\:\:\:\:\:\:\:\:\:\:\:\:\:\:\:\:\:\:\:\:\:\:\:\:\:\:\:\:\:\:\:\:\:\:\:\: \:\:\:\:\:\:\:\:\:\:\:\:\:\:\:\:\:\:\:\:\:\:\:\:\:\:\:\:\:\:\:\:\:\:\:\:\: 
\end{displaymath}
The height of the proof of the premise of application of the rule \((K_{I}\Rightarrow)\) reduced to \(<h-1>\). By the induction hypothesis the sequent \((t:A, s:K_{I}A, s\overset{I}{\sim}t, \Pi'_{atomic}, \Gamma'\Rightarrow\Delta, \Lambda_{atomic})\) is provable with the height \(h'\), where \(h' \leq h-1\). The sequent of the lemma is proved by applying the rule \((K_{I}\Rightarrow)\):
\begin{displaymath}
\frac{t:A, s:K_{I}A, s\overset{I}{\sim}t, \Pi'_{atomic}, \Gamma'\Rightarrow\Delta, \Lambda_{atomic}}{s:K_{I}A, s\overset{I}{\sim}t, \Pi'_{atomic}, \Gamma'\Rightarrow\Delta, \Lambda_{atomic}}\:\:( \: K_{I}\Rightarrow) \:\:\:\:\:\:\:\:\:\:\:\:\:\:\:\:\:\:\:\:\:\:\:\:\:\:\:\:\:\:\:\:\:\:\:\:\:\:\:\:\:\:\:\:\:\:\:\:\:\:\:\:\:\:\:\:\:\:\:\:\:\:\:\:\:\:\:\: \:\:\:\:\:\:\:\:\:\:\:\:\:\:\:\:\:\:\:\:\:\:\:\:\:\:\:\:\:\:\:\:\:\:\:\:\: 
\end{displaymath}
Other cases are prooved in a similar way.
\vspace{2mm}
\item Any formula of the principal pair is not in \(\Pi_{atomic}\).
\begin{displaymath}
\frac{t:A, s:K_{I}A, s\overset{I}{\sim}t, \Pi_{atomic}, \Pi_{atomic}, \Gamma'\Rightarrow\Delta, \Lambda_{atomic}, \Lambda_{atomic}}{s:K_{I}A, s\overset{I}{\sim}t, \Pi_{atomic}, \Pi_{atomic}, \Gamma'\Rightarrow\Delta, \Lambda_{atomic}, \Lambda_{atomic}}\:\:( \: K_{I}\Rightarrow) \:\:\:\:\:\:\:\:\:\:\:\:\:\:\:\:\:\:\:\:\:\:\:\:\:\:\:\:\:\:\:\:\:\:\:\:\:\:\:\:\:\:\:\:\:\:\:\:\:\:\:\:\:\:\:\:\:\:\:\:\:\:\:\:\:\:\:\: \:\:\:\:\:\:\:\:\:\:\:\:\:\:\:\:\:\:\:\:\:\:\:\:\:\:\:\:\:\:\:\:\:\:\:\:\: 
\end{displaymath}
By the induction hypothesis the sequent \((t:A, s:K_{I}A, s\overset{I}{\sim}t, \Pi_{atomic}, \Gamma'\Rightarrow\Delta, \Lambda_{atomic})\) is provable with the height \(h'\), where \(h' \leq h-1\). The sequent of the lemma is proved by applying the rule \((K_{I}\Rightarrow)\):
\begin{displaymath}
\frac{t:A, s:K_{I}A, s\overset{I}{\sim}t, \Pi_{atomic}, \Gamma'\Rightarrow\Delta, \Lambda_{atomic}}{s:K_{I}A, s\overset{I}{\sim}t, \Pi_{atomic}, \Gamma'\Rightarrow\Delta, \Lambda_{atomic}}\:\:( \: K_{I}\Rightarrow) \:\:\:\:\:\:\:\:\:\:\:\:\:\:\:\:\:\:\:\:\:\:\:\:\:\:\:\:\:\:\:\:\:\:\:\:\:\:\:\:\:\:\:\:\:\:\:\:\:\:\:\:\:\:\:\:\:\:\:\:\:\:\:\:\:\:\:\: \:\:\:\:\:\:\:\:\:\:\:\:\:\:\:\:\:\:\:\:\:\:\:\:\:\:\:\:\:\:\:\:\:\:\:\:\: 
\end{displaymath}
\end{itemize}
\item The cases of the remaining rules are considered similarly.\\
\end{itemize}
\end{proof}

\begin{lemmaA}[Substitution]\label{Lemma:Substitution}
If a sequent \((\Gamma \Rightarrow \Delta)\) is provable in GS-LCK, then the sequent \((\Gamma (t/s) \Rightarrow \Delta (t/s))\) is also provable with the same bound of the height of the proof in GS-LCK.
\end{lemmaA}
\begin{proof}\hspace*{1mm}\\
Lemma \ref{Lemma:Substitution} is proved by induction on the height \(<h>\) of the proof of the sequent \((\Gamma \Rightarrow \Delta)\).

\(<h=1>\)

If the sequent \((\Gamma \Rightarrow \Delta)\) is an axiom, then the sequent \((\Gamma (t/s) \Rightarrow \Delta (t/s))\) is an axiom as well.

\(<h>1>\)
\begin{itemize}[itemsep=2pt,topsep=0pt,parsep=0pt,partopsep=0pt]
\item The rule \((\Rightarrow K_{I})\) was applied in the last step of the proof of the sequent.
\begin{displaymath}
\frac{s\overset{I}{\sim}t, \Gamma\Rightarrow \Delta, t:A}{\Gamma\Rightarrow \Delta, s:K_{I}A}\:\:(\Rightarrow K_{I}) 
\end{displaymath}
\begin{itemize}[label=$-$,itemsep=2pt,topsep=0pt,parsep=0pt,partopsep=0pt]
\item Substitution \((l/z)\).

\vspace{10 pt}
By the induction hypothesis the sequent \((s\overset{I}{\sim}t, \Gamma (l/z) \Rightarrow \Delta (l/z), t:A)\) is provable with the height \(h'\), where \(h' \leq h-1\). The sequent of the lemma is proved by applying the rule \((\Rightarrow K_{I})\):
\begin{displaymath}
\frac{s\overset{I}{\sim}t, \Gamma (l/z) \Rightarrow \Delta (l/z), t:A}{\Gamma (l/z) \Rightarrow \Delta (l/z), s:K_{I}A}\:\:(\Rightarrow K_{I}) 
\end{displaymath}
\item Substitution \((l/t)\).

\vspace{10 pt}
There is no label \(t\) in the sequent \(\Gamma\Rightarrow \Delta, s:K_{I}A\) because of the requirement of the application of the rule \((\Rightarrow K_{I})\) that \(t\) is a new label.\\
\item Substitution \((l/s)\) and \(l \neq t\).

\vspace{10 pt}
By the induction hypothesis the sequent \((l\overset{I}{\sim}t, \Gamma (l/s) \Rightarrow \Delta (l/s), t:A)\) is provable with the height \(h'\), where \(h' \leq h-1\). The sequent of the lemma is proved by applying the rule \((\Rightarrow K_{I})\):
\begin{displaymath}
\frac{l\overset{I}{\sim}t, \Gamma (l/s) \Rightarrow \Delta (l/s), t:A}{\Gamma (l/s) \Rightarrow \Delta (l/s), l:K_{I}A}\:\:(\Rightarrow K_{I}) 
\end{displaymath}
\item Substitution \((l/s)\) and \(l = t\).

\vspace{10 pt}
By the induction hypothesis with substitution \((w/t)\), the sequent \((s\overset{I}{\sim}w, \Gamma \Rightarrow \Delta, w:A)\) is provable with the height \(h'\), where \(h' \leq h-1\). The label \(w\) is a new label absent in the sequent. By the inducion hypothesis with substitution \((l/s)\), the sequent \((l\overset{I}{\sim}w, \Gamma (l/s) \Rightarrow \Delta (l/s), w:A)\) is provable with the height \(h''\), where \(h'' \leq h-1\). The sequent of the lemma is proved by applying the rule \((\Rightarrow K_{I})\):
\begin{displaymath}
\frac{l\overset{I}{\sim}w, \Gamma (l/s) \Rightarrow \Delta (l/s), w:A}{\Gamma (l/s) \Rightarrow \Delta (l/s), l:K_{I}A}\:\:(\Rightarrow K_{I}) 
\end{displaymath}
\end{itemize}
\item The rule \((Ref)\) was applied in the last step of the proof of the sequent.
\begin{displaymath}
\frac{ s\overset{I}{\sim}s, \Gamma\Rightarrow\Delta}{\Gamma\Rightarrow\Delta}\:\:(Ref) \:\:
\end{displaymath}
\begin{itemize}[label=$-$,itemsep=2pt,topsep=0pt,parsep=0pt,partopsep=0pt]
\item Substitution \((s/t)\), and relational atom \(s\overset{I}{\sim}t\) is in \(\Gamma\).

\vspace{10 pt}
By the induction hypothesis, the sequent \((s\overset{I}{\sim}s, s\overset{I}{\sim}s, \Gamma (s/t) \Rightarrow \Delta (s/t))\) is provable with the height \(h'\), where \(h' \leq h-1\). The sequent of the lemma is proved by applying Lemma \ref{Lemma:AdmissibilityOfContractionAtomic}.\\
\item Other substitutions are considered in a similar way.\\
\end{itemize}
\item The cases of the remaining rules are considered similarly.\\
\end{itemize}
\end{proof}

\begin{theoremA}[Admissibility of weakening]\label{Theorem:AdmissibilityOfWeakening}
If a sequent \((\Gamma \Rightarrow \Delta)\) is provable in GS-LCK, then a sequent \((\Pi,\Gamma \Rightarrow \Delta, \Lambda)\) is provable with the same bound of the height of the proof in GS-LCK, too. \(\Pi, \Gamma, \Delta, \Lambda\) are any multisets of formulas.
\end{theoremA}
\begin{proof}\hspace*{1mm}\\
Theorem \ref{Theorem:AdmissibilityOfWeakening} is proved by induction on the height \(<h>\) of the proof of the sequent \((\Gamma \Rightarrow \Delta)\).

\(<h=1>\)

If the sequent \((\Gamma \Rightarrow \Delta)\) is an axiom, then the sequent \((\Pi,\Gamma \Rightarrow \Delta, \Lambda)\) is an axiom, as well.

\(<h>1>\)
\begin{itemize}[itemsep=2pt,topsep=0pt,parsep=0pt,partopsep=0pt]
\item The rule \((\Rightarrow K_{I})\) was applied in the last step of the proof of the sequent.
\begin{displaymath}
\frac{s\overset{I}{\sim}t, \Gamma\Rightarrow \Delta, t:A}{\Gamma\Rightarrow \Delta, s:K_{I}A}\:\:(\Rightarrow K_{I})
\end{displaymath}
\begin{itemize}[label=$-$,itemsep=2pt,topsep=0pt,parsep=0pt,partopsep=0pt]
\item A new label \(t\) for the application of the rule \((\Rightarrow K_{I})\) is in \(\Pi\) or \(\Lambda\).\\
By Lemma \ref{Lemma:Substitution}, the sequent \((s\overset{I}{\sim}t,\Gamma  \Rightarrow \Delta, t:A)\) with substitution \((l/t)\) is provable. By the induction hypothesis, the sequent \((s\overset{I}{\sim}l,\Pi,\Gamma  \Rightarrow \Delta, \Lambda, l:A)\) is provable with the height \(h'\), where \(h' \leq h-1\). Here \(l\) is a new label, absent in \(\Pi,\Gamma,\Delta\) and \(\Lambda\). The sequent of the theorem is proved by applying the rule \((\Rightarrow K_{I})\):
\begin{displaymath}
\frac{s\overset{I}{\sim}l,\Pi,\Gamma  \Rightarrow \Delta, \Lambda, l:A}{\Pi, \Gamma\Rightarrow \Delta, \Lambda, s:K_{I}A}\:\:(\Rightarrow K_{I})
\end{displaymath}
\vspace{2mm}
\item The new label \(t\) for application of the rule \((\Rightarrow K_{I})\) is absent in \(\Pi\) or \(\Lambda\).\\
By the induction hypothesis, the sequent \((s\overset{I}{\sim}t,\Pi,\Gamma  \Rightarrow \Delta, \Lambda, t:A)\) is provable with the height \(h'\), where \(h' \leq h-1\). The sequent of the theorem is proved by applying the rule \((\Rightarrow K_{I})\):
\begin{displaymath}
\frac{s\overset{I}{\sim}t,\Pi,\Gamma  \Rightarrow \Delta, \Lambda, t:A}{\Pi, \Gamma\Rightarrow \Delta, \Lambda, s:K_{I}A}\:\:(\Rightarrow K_{I})
\end{displaymath}
\vspace{2mm}
\end{itemize}
\item The cases of the remaining rules are considered similarly.\\
\end{itemize}
\end{proof}

\begin{theoremA}[Invertibility of rules]\label{Theorem:InvertibilityOfRules}
All the rules of GS-LCK are invertible with the same bound of the height of the proof.
\end{theoremA}
\begin{proof}\hspace*{1mm}\\
Theorem \ref{Theorem:InvertibilityOfRules} is proved for each rule separately.\\
The rule \((K_{I}\Rightarrow)\)
\begin{displaymath}
\frac{t:A, s:K_{I}A, s\overset{I}{\sim}t, \Gamma\Rightarrow\Delta}{s:K_{I}A, s\overset{I}{\sim}t, \Gamma\Rightarrow\Delta}\:\:( \: K_{I}\Rightarrow)
\end{displaymath}
Invertibility is proved by induction on the height \(<h>\) of the proof of the sequent of the conclusion of the rule  \((K_{I}\Rightarrow)\).

\(<h=1>\)

If the sequent \((s:K_{I}A, s\overset{I}{\sim}t, \Gamma\Rightarrow\Delta)\) is an axiom, then the sequent \((t:A, s:K_{I}A, s\overset{I}{\sim}t, \Gamma\Rightarrow\Delta)\) is an axiom, too.

\(<h>1>\)
\begin{itemize}[itemsep=2pt,topsep=0pt,parsep=0pt,partopsep=0pt]
\item The formula \(s\overset{I}{\sim}t\) is the principal formula.\\
\begin{itemize}[label=$-$,itemsep=2pt,topsep=0pt,parsep=0pt,partopsep=0pt]
\item The rule \((\:Sub(o^{r}) \Rightarrow)\) was applied in the last step of the proof of the sequent.\\
\begin{displaymath}
\frac{s:o^{r}, s:K_{I}A, s\overset{I}{\sim}t,  t:o^{r}, \Gamma' \Rightarrow  \Delta}{s:K_{I}A, s\overset{I}{\sim}t,  t:o^{r}, \Gamma' \Rightarrow\Delta}\:\:(\:Sub(o^{r}) \Rightarrow) \:\:\:\:\:\:\:\:\:\:\:\:\:\:\:\:\:\:\:\:\:\:\:\:\:\:\:\:\:\:\:\:\:\:\:\:\:\:\:\:\:\:\:\:\:\:\:\:\:\:\:\:\:\:\:\:\:\:\:\:\:\:\:\:
\end{displaymath}\\
By the induction hypothesis, the sequent \((t:A, s:o^{r}, s:K_{I}A, s\overset{I}{\sim}t,  t:o^{r}, \Gamma' \Rightarrow  \Delta)\) is provable with the height \(h'\), where \(h' \leq h-1\). The sequent of the premise of the rule \((K_{I}\Rightarrow)\) is proved by applying the rule \((\:Sub(o^{r}) \Rightarrow)\):
\begin{displaymath}
\frac{t:A, s:o^{r}, s:K_{I}A, s\overset{I}{\sim}t,  t:o^{r}, \Gamma' \Rightarrow  \Delta}{t:A, s:K_{I}A, s\overset{I}{\sim}t,  t:o^{r}, \Gamma' \Rightarrow  \Delta}\:\:(\:Sub(o^{r}) \Rightarrow) \:\:\:\:\:\:\:\:\:\:\:\:\:\:\:\:\:\:\:\:\:\:\:\:\:\:\:\:\:\:\:\:\:\:\:\:\:\:\:\:\:\:\:\:\:\:\:\:\:\:\:\:\:\:\:\:\:\:\:\:\:\:\:\:
\end{displaymath}
\item For rules \((K_{I}\Rightarrow)\), \((Trans)\), \((Eucl)\), \((Mon)\) in a similar way.\\
\end{itemize}
\item The case where the formula \(s:K_{I}A\) is the principal formula and the case where formulas \(s\overset{I}{\sim}t\) and \(s:K_{I}A\) both are not principal formulas are considered similarly.
\end{itemize}
Invertibility of the remaining rules is proved in a similar way.\\
\end{proof}

\begin{theoremA}[Admissibility of contraction]\label{Theorem:AdmissibilityOfContraction}
If a sequent \((\Pi,\Pi,\Gamma \Rightarrow \Delta, \Lambda, \Lambda)\) is provable in GS-LCK, then sequent \((\Pi,\Gamma \Rightarrow \Delta, \Lambda)\) is provable with the same bound of the height of the proof in GS-LCK, too. \(\Pi,\Gamma, \Delta,\Lambda\) are any multisets of formulas. 
\end{theoremA}
\begin{proof}\hspace*{1mm}\\
Theorem \ref{Theorem:AdmissibilityOfContraction} is proved by induction on the ordered tuple pair \(<c,h>\), where \(c\) is the sum of complexity of all the formulas in \(\Pi\) and \(\Lambda\), and \(h\) is the height of the proof of the sequent \((\Pi,\Pi,\Gamma \Rightarrow \Delta, \Lambda, \Lambda)\).

\(<c \geq 1, h=1>\)

If the sequent \((\Pi,\Pi,\Gamma \Rightarrow \Delta, \Lambda, \Lambda)\) is an axiom, then the sequent \((\Pi,\Gamma \Rightarrow \Delta, \Lambda)\) is an axiom, too.

\(<c \geq 1, h>1>\)
\begin{itemize}[itemsep=2pt,topsep=0pt,parsep=0pt,partopsep=0pt]
\item The rule \((\neg\Rightarrow)\) was applied in the last step of the proof of the sequent.\\
\begin{itemize}[label=$-$,itemsep=2pt,topsep=0pt,parsep=0pt,partopsep=0pt]
\item The principal formula is in \(\Pi\).
\begin{displaymath}
\frac{s:\neg A, \Pi', \Pi', \Gamma \Rightarrow \Delta, \Lambda, \Lambda, s:A}{s:\neg A, s:\neg A, \Pi', \Pi', \Gamma \Rightarrow \Delta, \Lambda, \Lambda}\:\:(\neg\Rightarrow) \:\:\:\:\:\:\:\:\:\:\:\:\:\:\:\:\:\:\:\:\:\:\:\:\:\:\:\:\:\:\:\:\:\:\:\:\:\:\:\:\:\:\:\:\:\:\:\:\:\:\:\:\:\:\:\:\:\:\:\:\:\:\:\:\:\:\:\:\:\:\:\:\:\:\:\:\:\:\:\:\:\:\:\:\:\:\:\:\:\:\:\:\:\:\:\:\:\:\:\:
\end{displaymath}
By invertibility of the rule \((\neg\Rightarrow)\), the sequent \((\Pi', \Pi', \Gamma \Rightarrow \Delta, \Lambda, \Lambda, s:A, s:A)\) is provable. The value of the ordered tuple pair has reduced to \(<c-1, h>\). By the induction hypothesis, the sequent \((\Pi', \Gamma \Rightarrow \Delta, \Lambda, s:A)\) is provable with the height \(h'\), where \(h' \leq h-1\). The sequent of the theorem is proved by applying the rule \((\neg\Rightarrow)\):
\begin{displaymath}
\frac{\Pi', \Gamma \Rightarrow \Delta, \Lambda, s:A}{s:\neg A, \Pi', \Gamma \Rightarrow \Delta, \Lambda}\:\:(\neg\Rightarrow) \:\:\:\:\:\:\:\:\:\:\:\:\:\:\:\:\:\:\:\:\:\:\:\:\:\:\:\:\:\:\:\:\:\:\:\:\:\:\:\:\:\:\:\:\:\:\:\:\:\:\:\:\:\:\:\:\:\:\:\:\:\:\:\:\:\:\:\:\:\:\:\:\:\:\:\:\:\:\:\:\:\:\:\:\:\:\:\:\:\:\:\:\:\:\:\:\:\:\:\:
\end{displaymath}
\item The principal formula is absent in \(\Pi\).
\begin{displaymath}
\frac{\Pi, \Pi, \Gamma \Rightarrow \Delta, \Lambda, \Lambda, s:A}{s:\neg A, \Pi, \Pi, \Gamma \Rightarrow \Delta, \Lambda, \Lambda}\:\:(\neg\Rightarrow) \:\:\:\:\:\:\:\:\:\:\:\:\:\:\:\:\:\:\:\:\:\:\:\:\:\:\:\:\:\:\:\:\:\:\:\:\:\:\:\:\:\:\:\:\:\:\:\:\:\:\:\:\:\:\:\:\:\:\:\:\:\:\:\:\:\:\:\:\:\:\:\:\:\:\:\:\:\:\:\:\:\:\:\:\:\:\:\:\:\:\:\:\:\:\:\:\:\:\:\:
\end{displaymath}
By the induction hypothesis, the sequent \((\Pi, \Gamma \Rightarrow \Delta, \Lambda, s:A)\) is provable with the height \(h'\), where \(h' \leq h-1\). The sequent of the theorem is proved by applying the rule \((\neg\Rightarrow)\):
\begin{displaymath}
\frac{\Pi, \Gamma \Rightarrow \Delta, \Lambda, s:A}{s:\neg A, \Pi, \Gamma \Rightarrow \Delta, \Lambda}\:\:(\neg\Rightarrow) \:\:\:\:\:\:\:\:\:\:\:\:\:\:\:\:\:\:\:\:\:\:\:\:\:\:\:\:\:\:\:\:\:\:\:\:\:\:\:\:\:\:\:\:\:\:\:\:\:\:\:\:\:\:\:\:\:\:\:\:\:\:\:\:\:\:\:\:\:\:\:\:\:\:\:\:\:\:\:\:\:\:\:\:\:\:\:\:\:\:\:\:\:\:\:\:\:\:\:\:
\end{displaymath}
\end{itemize}
\item The cases of the remaining rules are considered similarly.\\
\end{itemize}
\end{proof}

\begin{theoremA}[Admissibility of cut]\label{Theorem:AdmissibilityOfCut}
If sequents \((\Gamma \Rightarrow \Delta, F)\) and \((F, \Pi \Rightarrow \Lambda)\) are provable in GS-LCK, then sequent \((\Pi,\Gamma \Rightarrow \Delta, \Lambda)\) is provable in GS-LCK too. \(F\) is any formula and \(\Pi,\Gamma, \Delta,\Lambda\) are any multisets of formulas. 
\end{theoremA}
\begin{proof}\hspace*{1mm}\\
Theorem \ref{Theorem:AdmissibilityOfCut} is proved by induction on the ordered tuple pair \(<c,h>\), where \(c\) is the complexity of formula \(F\), and \(h\) is the sum of heights of the proof of the sequents \((\Gamma \Rightarrow \Delta, F)\) and \((F, \Pi \Rightarrow \Lambda)\).

\(<c \geq 1, h=2>\)

The sequents \((\Gamma \Rightarrow \Delta, F)\) and \((F, \Pi \Rightarrow \Lambda)\) are the axioms. If formula \(F\) is not principal in one at least of the sequents, then \((\Pi,\Gamma \Rightarrow \Delta, \Lambda)\) is an axiom. If formula \(F\) is principal in both sequents, then \(F\) should be in \(\Gamma\) and \(\Delta\) or only in \(\Gamma\) (the case where the axiom is of type \(s:o^{r_{1}},s:o^{r_{2}},\Gamma\Rightarrow \Delta\)). Therefore the sequent \((\Pi,\Gamma \Rightarrow \Delta, \Lambda)\) is also an axiom.

\(<c \geq 1, h>2>\)
\begin{itemize}[itemsep=2pt,topsep=0pt,parsep=0pt,partopsep=0pt]
\item Formula \(F\) is not principal in the sequent \((\Gamma \Rightarrow \Delta, F)\).\\
\begin{itemize}[label=$-$,itemsep=2pt,topsep=0pt,parsep=0pt,partopsep=0pt]
\item The rule \((Sub(o^{r}) \Rightarrow)\) was applied in the last step of the proof of the sequent \((\Gamma \Rightarrow \Delta, F)\).
\begin{displaymath}
\frac{s:o^{r}, t:o^{r}, s\overset{N}{\sim}t, \Gamma\Rightarrow\Delta, F}{ t:o^{r}, s\overset{N}{\sim}t, \Gamma\Rightarrow\Delta, F}\:\:( \: Sub(o^{r}) \Rightarrow)\:\:\:\:\:\:\:\:\:\:\:\:\:\:\:\:\:\:\:\:\:\:\:\:\:\:\:\:\:\:\:\:\:\:\:\:\:\:\:\:\:\:\:\:\:\:\:\:\:\:\:\:\:\:\:\:\:\:\:\:\:\:\:\:\:\:\:\:\:\:\:\:\:\:\:\:\:\:\:\:\:\:\:\:\:\:\:\:\:\:\:\:\:\:\:\:\:\:\:\:\:\:\:\:\:\:\:\:\:\:\:\:\:\:
\end{displaymath}
By the induction hypothesis, the sequent \((s:o^{r}, t:o^{r}, s\overset{N}{\sim}t,\Pi,\Gamma \Rightarrow \Delta, \Lambda)\) is provable. The sequent of the theorem is proved by applying the rule \((Sub(o^{r}) \Rightarrow)\):
\begin{displaymath}
\frac{s:o^{r}, t:o^{r}, s\overset{N}{\sim}t,\Pi,\Gamma \Rightarrow \Delta, \Lambda}{ t:o^{r}, s\overset{N}{\sim}t,\Pi,\Gamma \Rightarrow \Delta, \Lambda}\:\:( \: Sub(o^{r}) \Rightarrow)\:\:\:\:\:\:\:\:\:\:\:\:\:\:\:\:\:\:\:\:\:\:\:\:\:\:\:\:\:\:\:\:\:\:\:\:\:\:\:\:\:\:\:\:\:\:\:\:\:\:\:\:\:\:\:\:\:\:\:\:\:\:\:\:\:\:\:\:\:\:\:\:\:\:\:\:\:\:\:\:\:\:\:\:\:\:\:\:\:\:\:\:\:\:\:\:\:\:\:\:\:\:\:\:\:\:\:\:\:\:\:\:\:\:
\end{displaymath}
\item For applications of other rules in a similar way.\\
\end{itemize}
\item Formula \(F\) is not principal in the sequent \((F, \Pi \Rightarrow \Lambda)\).\\
The case is considered in a similar way.\\
\item Formula \(F\) is principal in both sequents \((\Gamma \Rightarrow \Delta, F)\) and \((F, \Pi \Rightarrow \Lambda)\).\\
\begin{itemize}[label=$-$,itemsep=2pt,topsep=0pt,parsep=0pt,partopsep=0pt]
\item The sequent \((\Gamma \Rightarrow \Delta, F)\) is an axiom and the rule \((OE)\) was applied in the last step of the proof of the sequent \((F, \Pi \Rightarrow \Lambda)\).
\begin{displaymath}
s:o_{1}^{r_{o_{1}}},\Gamma\Rightarrow \Delta, s:o_{1}^{r_{o_{1}}} \:\:\:\:\:\:\:\:\:\: \:\:\:\:\:\:\:\:\:\: \:\:\:\:\:\:\:\:\:\:\:\:\:\:\:\:\:\:\:\: \:\:\:\:\:\:\:\:\:\: \:\:\:\:\:\:\:\:\:\:\:\:\:\:\:\:\:\:\:\: \:\:\:\:\:\:\:\:\:\: \:\:\:\:\:\:\:\:\:\:\:\:\:\:\:\:\:\:\:\: \:\:\:\:\:\:\:\:\:\: \:\:\:\:\:\:\:\:\:\:\:\:\:\:\:\:\:\:\:\: \:\:\:\:\:\:\:\:\:\: \:\:\:\:\:\:\:\:\:\:
\end{displaymath}
\begin{displaymath}
\frac{s\overset{I}{\sim}t,s:o_{1}^{r_{o_{1}}}, \{s:o^{r_{o}} \}_{o\in \{O_{I}\backslash o_{1} \}}, \{t:o^{r_{o}} \}_{o\in O_{I}}, \Pi\Rightarrow  \Lambda}{s:o_{1}^{r_{o_{1}}}, \{s:o^{r_{o}} \}_{o\in \{O_{I}\backslash o_{1} \}}, \{t:o^{r_{o}} \}_{o\in O_{I}}, \Pi\Rightarrow  \Lambda}\:(OE) \:\:\:\:\:\:\:\: \:\:\:\:\:\:\:\:\:\: \:\:\:\:\:\:\:\:\:\:\:\:\:\:\:\:\:\:\:\: \:\:\:\:\:\:\:\:\:\: \:\:\:\:\:\:\:\:\:\:
\end{displaymath}
By the induction hypothesis, the sequent \((s:o_{1}^{r_{o_{1}}},s\overset{I}{\sim}t,\{s:o^{r_{o}} \}_{o\in \{O_{I}\backslash o_{1} \}}, \{t:o^{r_{o}} \}_{o\in O_{I}}, \Pi, \Gamma\Rightarrow \Delta,  \Lambda)\) is provable. The sequent of the theorem is proved by applying the rule \((OE)\):
\begin{displaymath}
\frac{s:o_{1}^{r_{o_{1}}},s\overset{I}{\sim}t,\{s:o^{r_{o}} \}_{o\in \{O_{I}\backslash o_{1} \}}, \{t:o^{r_{o}} \}_{o\in O_{I}}, \Pi, \Gamma\Rightarrow \Delta,  \Lambda}{s:o_{1}^{r_{o_{1}}},\{s:o^{r_{o}} \}_{o\in \{O_{I}\backslash o_{1} \}}, \{t:o^{r_{o}} \}_{o\in O_{I}}, \Pi, \Gamma\Rightarrow \Delta,  \Lambda}\:(OE) \:\:\:\:\:\:\:\: \:\:\:\:\:\:\:\:\:\: \:\:\:\:\:\:\:\:\:\:\:\:\:\:\:\:\:\:\:\: \:\:\:\:\:\:\:\:\:\: \:\:\:\:\:\:\:\:\:\:
\end{displaymath}
\item The cases of the remaining rules are considered similarly.
\end{itemize}
\end{itemize}
\end{proof}

\section{Proof of completeness of GS-LCK}
\label{sec:ProofOfCompletenessOfGSLCK}

\begin{theoremA}[Completeness of GS-LCK]\label{theoremA:1}
If formula \(A\) is valid with respect to correlation models over \((R,\Sigma,\vec O)\), then sequent \((\Rightarrow s:A)\) is provable in GS-LCK.
\end{theoremA}
\begin{proof}\hspace*{1mm}\\
The Hilbert style proof system HS-LCK for logic of correlated knowledge is complete. Showing the provability of all valid formulas of HS-LCK in GS-LCK, the completeness of GS-LCK is proved. Theorem \ref{theoremA:1} is proved by induction on the number of steps \(<NSteps>\), used to prove formula \(A\) in HS-LCK.

\(<NSteps=1>\)

Formula \(A\) is an axiom of calculus HS-LCK.

\begin{itemize}[itemsep=2pt,topsep=0pt,parsep=0pt,partopsep=0pt]
\item The axiom "H4. \(K_{I}(A\rightarrow B)\rightarrow (K_{I} A \rightarrow K_{I} B) \)", was used.
\begin{prooftree}
\AxiomC{\(t:A, ... \Rightarrow  t:B, t:A\)}
\AxiomC{\(t:B, t:A, ... \Rightarrow  t:B\)}
\RightLabel{\((\rightarrow \Rightarrow)\)}
\BinaryInfC{\( t:A\rightarrow B, t:A, s\overset{I}{\sim}t, s:K_{I}(A\rightarrow B), s:K_{I} A \Rightarrow  t:B\)}
\RightLabel{\((K_{I} \Rightarrow)\)}
\UnaryInfC{\( t:A, s\overset{I}{\sim}t, s:K_{I}(A\rightarrow B), s:K_{I} A \Rightarrow  t:B\)}
\RightLabel{\((K_{I} \Rightarrow)\)}
\UnaryInfC{\( s\overset{I}{\sim}t, s:K_{I}(A\rightarrow B), s:K_{I} A \Rightarrow  t:B\)}
\RightLabel{\((\Rightarrow K_{I})\)}
\UnaryInfC{\( s:K_{I}(A\rightarrow B), s:K_{I} A \Rightarrow  s:K_{I} B\)}
\RightLabel{\((\Rightarrow \rightarrow)\)}
\UnaryInfC{\( s:K_{I}(A\rightarrow B)\Rightarrow s:K_{I} A \rightarrow K_{I} B\)}
\RightLabel{\((\Rightarrow \rightarrow)\)}
\UnaryInfC{\(\Rightarrow s:K_{I}(A\rightarrow B)\rightarrow (K_{I} A \rightarrow K_{I} B)\)}
\end{prooftree}
\hspace*{1mm}
\item The axiom "H8. \(K_{I}A\rightarrow K_{J}A,  \:\:\text{when}\:\: I\subseteq J\)", was used.
\begin{prooftree}
\AxiomC{\(t:A, s\overset{I}{\sim}t, s\overset{J}{\sim}t, s:K_{I}A \Rightarrow  t:A\)}
\RightLabel{\((K_{I} \Rightarrow)\)}
\UnaryInfC{\(s\overset{I}{\sim}t, s\overset{J}{\sim}t, s:K_{I}A \Rightarrow  t:A\)}
\RightLabel{\((Mon)\)}
\UnaryInfC{\(s\overset{J}{\sim}t, s:K_{I}A \Rightarrow  t:A\)}
\RightLabel{\((\Rightarrow K_{J})\)}
\UnaryInfC{\(s:K_{I}A \Rightarrow  s:K_{J}A\)}
\RightLabel{\((\Rightarrow \rightarrow)\)}
\UnaryInfC{\(\Rightarrow s:K_{I}A\rightarrow K_{J}A\)}
\end{prooftree}
\hspace*{1mm}
\item The axiom "H12. \(o^{r}_{I} \rightarrow K_{I} o^{r}_{I}\)", was used.
\begin{prooftree}
\AxiomC{\(t:o^{r}_{I}, s\overset{I}{\sim}t, s:o^{r}_{I} \Rightarrow  t:o^{r}_{I}\)}
\RightLabel{\(( \: Sub(o^{r}) \Rightarrow)\)}
\UnaryInfC{\(s\overset{I}{\sim}t, s:o^{r}_{I} \Rightarrow  t:o^{r}_{I}\)}
\RightLabel{\((\Rightarrow K_{I})\)}
\UnaryInfC{\(s:o^{r}_{I} \Rightarrow  s:K_{I} o^{r}_{I}\)}
\RightLabel{\((\Rightarrow \rightarrow)\)}
\UnaryInfC{\(\Rightarrow s:o^{r}_{I} \rightarrow K_{I} o^{r}_{I}\)}
\end{prooftree}
\hspace*{1mm}
\item The axiom "H13. \((\underset{o\in O_{I}}{\wedge} o^{r_{o}} \wedge K_{I}A)\rightarrow  K_{\emptyset}(\underset{o\in O_{I}}{\wedge} o^{r_{o}} \rightarrow A),\:\:\text{when}\:\:I \subset N\)", was used.
\begin{prooftree}
\AxiomC{\(t:A, s\overset{I}{\sim}t, t:\underset{o\in O_{I}}{\wedge} o^{r_{o}}, s\overset{\emptyset}{\sim}t,s:\underset{o\in O_{I}}{\wedge} o^{r_{o}}, s:K_{I}A \Rightarrow t:A\)}
\RightLabel{\((K_{I} \Rightarrow)\)}
\UnaryInfC{\(s\overset{I}{\sim}t, t:\underset{o\in O_{I}}{\wedge} o^{r_{o}}, s\overset{\emptyset}{\sim}t,s:\underset{o\in O_{I}}{\wedge} o^{r_{o}}, s:K_{I}A \Rightarrow t:A\)}
\RightLabel{\((OE)\)}
\UnaryInfC{\(t:\underset{o\in O_{I}}{\wedge} o^{r_{o}}, s\overset{\emptyset}{\sim}t,s:\underset{o\in O_{I}}{\wedge} o^{r_{o}}, s:K_{I}A \Rightarrow t:A\)}
\RightLabel{\((\Rightarrow \rightarrow)\)}
\UnaryInfC{\(s\overset{\emptyset}{\sim}t,s:\underset{o\in O_{I}}{\wedge} o^{r_{o}}, s:K_{I}A \Rightarrow t:\underset{o\in O_{I}}{\wedge} o^{r_{o}} \rightarrow A\)}
\RightLabel{\((\wedge \Rightarrow)\)}
\UnaryInfC{\(s\overset{\emptyset}{\sim}t,s:\underset{o\in O_{I}}{\wedge} o^{r_{o}} \wedge K_{I}A \Rightarrow t:\underset{o\in O_{I}}{\wedge} o^{r_{o}} \rightarrow A\)}
\RightLabel{\((\Rightarrow K_{\emptyset})\)}
\UnaryInfC{\(s:\underset{o\in O_{I}}{\wedge} o^{r_{o}} \wedge K_{I}A \Rightarrow s:K_{\emptyset}(\underset{o\in O_{I}}{\wedge} o^{r_{o}} \rightarrow A)\)}
\RightLabel{\((\Rightarrow \rightarrow)\)}
\UnaryInfC{\(\Rightarrow s:(\underset{o\in O_{I}}{\wedge} o^{r_{o}} \wedge K_{I}A)\rightarrow  K_{\emptyset}(\underset{o\in O_{I}}{\wedge} o^{r_{o}} \rightarrow A)\)}
\end{prooftree}
\item The remaining axioms are considered in a similar way.\\
\end{itemize}
\(<NSteps>1>\)

One of the rules \((Modus\:ponens)\) or \((K_{I} - necessitation)\) of calculus HS-LCK was applied in the last step of the proof of the formula.

\begin{itemize}[itemsep=2pt,topsep=0pt,parsep=0pt,partopsep=0pt]
\item The rule \((Modus\:ponens)\) was applied.
\begin{displaymath}
\frac{A,A\rightarrow B}{B} \:\:(Modus\:ponens)
\end{displaymath}
By the induction hypothesis, sequents \((\Rightarrow s:A)\) and \((\Rightarrow s:A\rightarrow B)\) are provable in GS-LCK. By invertibility of the rule \((\Rightarrow \rightarrow)\), the sequent \((s:A\Rightarrow  s:B)\) is provable. The sequent \((\Rightarrow  s:B)\) of the theorem is proved by applying Theorem \ref{Theorem:AdmissibilityOfCut} "Admissibility of cut".\\

\item The rule \((K_{I} - necessitation)\) was applied.
\begin{displaymath}
\frac{A}{K_{I}A} \:\:(K_{I} - necessitation)
\end{displaymath}

By the induction hypothesis, the sequent \((\Rightarrow s:A)\) is provable in GS-LCK. By Lemma \ref{Lemma:Substitution} "Substition", the sequent \((\Rightarrow t:A)\) is provable. By Theorem \ref{Theorem:AdmissibilityOfWeakening} "Admissibility of weakening", the sequent \((s\overset{I}{\sim}t\Rightarrow t:A)\) is provable. The sequent of the theorem is proved by applying the rule \((\Rightarrow K_{I})\):
\begin{displaymath}
\frac{s\overset{I}{\sim}t\Rightarrow t:A}{\Rightarrow s:K_{I}A}\:\:(\Rightarrow K_{I}) \:\:\:\:\:\:\:\:\:\:\:\:\:\:\:\:\:\: 
\end{displaymath}
\end{itemize}
\end{proof}

\section{Automated proof search system GS-LCK-PROC}
\label{sec:DecidabilityOfLogicOfCorrelatedKnowledge}
Having sound and complete sequent calculus GS-LCK for logic of correlated knowledge we can model automated proof search system for LCK. GS-LCK-PROC is defined as procedure, which uses rules and axioms of sequent calculus GS-LCK. Principal formulas of the applications of the rules \((K_{I} \Rightarrow)\), \((K_{N} \Rightarrow)\) and \((\Rightarrow K_{I})\), and the chains of new appeared relational atoms of applications of the rule \((\Rightarrow K_{I})\) are saved in tables  \(TableLK\) and \(TableRK\).

\begin{definitionA}[Table TableLK]\label{definitionA:TableLK}
Table TableLK of the principal pairs of the applications of the rules \((K_{I} \Rightarrow)\) and \((K_{N} \Rightarrow)\):
\begin{center}
    \begin{tabular}{| c | c |}
    \hline
    \multicolumn{2}{|c|}{TableLK} \\ \hline
    Main formula & Relational atom \\ \hline
    \:\:\:\:\:\:\:\:\:\:\:\:\: \:\:\:\:\:\:\:\:\:\:\:\:\: \:\:\:\:\:\:\:\:\:\:\:\:\: & \:\:\:\:\:\:\:\:\:\:\:\:\: \:\:\:\:\:\:\:\:\:\:\:\:\: \:\:\:\:\:\:\:\:\:\:\:\:\: \\ \hline
     & \\ \hline
     & \\ \hline
    \end{tabular}
\end{center}
\end{definitionA}
\begin{exampleA}\label{definitionA:TableLK}
Example of TableLK:
\begin{center}
    \begin{tabular}{| l | l |}
    \hline
    \multicolumn{2}{|c|}{TableLK} \\ \hline
    \:\:\:\:\:\:\:\:Main formula\:\:\:\:\:\:\:\: & \:\:\:\:Relational atom \:\:\:\:\\ \hline
    \:\:\(s:K_{I}A\)   &   \:\:\(s\overset{I}{\sim}t\) \\ \hline
    \:\:\(l:K_{I}B\) & \:\:\(l\overset{I}{\sim}z\) \\ \hline
    \end{tabular}
\end{center}
\end{exampleA}
\begin{definitionA}[Table TableRK]\label{definitionA:TableRK}
Table TableRK of the principal formulas and chains of new appeared relational atoms of the applications of the rule \((\Rightarrow K_{I})\):
\begin{center}
    \begin{tabular}{| l | l | l | l |}
    \hline
    \multicolumn{4}{|c|}{TableRK} \\ \hline
   \:\: Main formula \:\:\:\:& \:\:Chain of the relational atoms \:\:& \:\:Length of chain \:\: & \: Max \: \\ \hline
     &  &  &  \\  \cline{2-4}
     &  &  &  \\ \cline{2-4}
     & &   &  \\ \hline
     &  &  &  \\  \cline{2-4}
     &  &  &  \\ \cline{2-4}
     & &   &  \\ \hline
    \end{tabular}
\end{center}
where \(Max\) is the maximum length of the chain, defined by \(n(K_{I})+1\). Formula \(n(K_{I})\) denotes the number of negative occurences of knowledge operator \(K_{I}\) in a sequent. 
\end{definitionA}
\begin{exampleA}\label{definitionA:TableRK}
Example of TableRK:
\begin{center}
    \begin{tabular}{| l | l | c | c |}
    \hline
    \multicolumn{4}{|c|}{TableRK} \\ \hline
   \:\: Main formula \:\:& \:\:Chain of the relational atoms \:\:& \:\:Length of chain \:\: & \: Max \: \\ \hline
\:\:\(s,s_{1},s_{2},w_{1}:K_{I}A\) 
     & \:\:\(s\overset{I}{\sim}s_{1}, s_{1}\overset{I}{\sim}s_{2}, s_{2}\overset{I}{\sim}s_{3}\) & \:\:3  & 5\\ \cline{2-4}
     & \:\:\(s\overset{I}{\sim}t_{1}\) & \:\:1 & 5 \\  \cline{2-4}
     & \:\:\(s\overset{I}{\sim}w_{1}, w_{1}\overset{I}{\sim}w_{2}\) & \:\:2 & 5 \\ \hline
\:\:\(z,z_{1}:K_{J}B\) 
     & \:\:\(z\overset{J}{\sim}z_{1}, z_{1}\overset{J}{\sim}z_{2}\) & \:\:2 & 7 \\ \hline
    \end{tabular}
\end{center}
\end{exampleA}

\begin{definitionA}[Procedure of the proof search]\label{definitionA:1}
Procedure GS-LCK-PROC of the proof search in the sequent calculus GS-LCK:\\\\
Initialisation:
\begin{itemize}[itemsep=2pt,topsep=0pt,parsep=0pt,partopsep=0pt]
\item Define set \(N\) of agents, tuple of sets \(\vec O = (O_{a_{1}},...,O_{a_{n}})\) of possible observations and result structure \((R,\Sigma)\).
\item Initialise the tables \(TableLK\) and \(TableRK\) by setting Max values to \((n(K_{I})+1)\), the length of the chain to 0 and the other cells leaving empty.
\item Set Output = False.
\end{itemize}
\hspace*{1mm}\\
PROCEDURE GS-LCK-PROC(Sequent, TableLK, TableRK, Output)\\\\
BEGIN\\
\begin{enumerate}[itemsep=6pt,topsep=0pt,parsep=0pt,partopsep=0pt]
\item Check if the sequent is the axiom. If the sequent is the axiom, set \(Output = True\) and go to step Finish.
\item If possible, apply any of the rules \((\neg\Rightarrow), (\Rightarrow\neg), (\Rightarrow\vee), (\wedge\Rightarrow), (\Rightarrow\rightarrow)\) and go to step 1.
\item If possible, apply any of the rules \((\vee\Rightarrow), (\Rightarrow\wedge)\) or \((\rightarrow\Rightarrow)\) and call procedure GS-LCK-PROC() for the premises of the application: \\\\
Output1 = False;\\
Output2 = False;\\\\
GS-LCK-PROC(Premise1, TableLK, TableRK, Output1);\\
GS-LCK-PROC(Premise2, TableLK, TableRK, Output2);\\\\

IF (Output1 == True) AND (Output2 == True) \\
THEN Set Output = True and go to Finish;\\
ELSE Set Output = False and go to Finish;\\

\item If possible to apply any of the rules \((K_{I} \Rightarrow)\) or \((K_{N} \Rightarrow)\), check if the principal pair is absent in the table TableLK. If it is absent, apply rule \((K_{I} \Rightarrow)\) or \((K_{N} \Rightarrow)\), add principal pair to TableLK and go to step 1.  
\item If possible to apply rule \((\Rightarrow K_{I})\), check if the principal formula is absent in the table TableRK and the length of the chain is lower than Max. If the principal formula is absent and the length of the chain is lower than Max, apply rule \((\Rightarrow K_{I})\), add principal formula and new relational atom to TableRK, increment the length of the chain by 1, and go to step 1.  
\item If possible, apply rule \((OYR)\) and call procedure GS-LCK-PROC() for the premises of the application: \\

For each k set Output(k) = False and call GS-LCK-PROC(Premise(k), TableLK, TableRK, Output(k)), where k is the index of the premise; \\

IF (for each k Output(k) == True) \\
THEN Set Output = True and go to Finish;\\
ELSE Set Output = False and go to Finish;\\
\item If possible, apply any of the rules \((\Rightarrow K_{N}), (OE), (CR), (Sub(p) \Rightarrow),(Sub(o^{r})\Rightarrow), (Ref), (Trans), (Eucl)\) or \((Mon)\) and go to step 1.
\item Finish.
\end{enumerate} 
END
\end{definitionA} 
Procedure GS-LCK-PROC gets the sequent, \(TableLK\), \(TableRK\), starting \(Output\) and returns "True", if the sequent is provable. Otherwise - "False", if it is not provable. Procedure is constructed in such a way, that it produces proofs, where number of applications of the knowledge rules of sequent calculus GS-LCK is finite. Also number of applications of other rules are bounded by requirements to rules and finite initial sets of agents, observations and results, which allows procedure to perform terminating proof search.  
\begin{lemmaA}[Permutation of the rule \((K_{I}\Rightarrow)\)]\label{Lemma:PermutationOfTheRuleLKI}
Rule \((K_{I}\Rightarrow)\) permutes down with respect to all rules of GS-LCK, except rules \((\Rightarrow K_{I})\) and \((OE)\). Rule \((K_{I}\Rightarrow)\) permutes down with rules \((\Rightarrow K_{I})\) and \((OE)\) in case the principal atom of \((K_{I}\Rightarrow)\) is not active in it.
\end{lemmaA}
\begin{proof}\hspace*{1mm}\\
The Lemma \ref{Lemma:PermutationOfTheRuleLKI} is proved in the same way as the Lemma 6.3. in \cite{Negri2005}. 
\end{proof}

\begin{lemmaA}[Number of applications of the rule \((K_{I}\Rightarrow)\)]\label{Lemma:MinimalApplicationsOfTheRuleLKI}
If a sequent \(S\) is provable in GS-LCK, then there exists the proof of \(S\) such that rule \((K_{I}\Rightarrow)\) is applied no more than once on the same pair of principal formulas on any branch.
\end{lemmaA}
\begin{proof}\hspace*{1mm}\\
The Lemma \ref{Lemma:MinimalApplicationsOfTheRuleLKI} is proved by induction on the number \(N\) of pairs of applications of rule \((K_{I}\Rightarrow)\) on the same branch with the same principal pair. \\\\
\(<N = 0>\) The proof of the lemma is obtained. \\\\
\(<N > 0>\) \\\\ 
We diminish the inductive paramater in the same way as in the proof of Corollary 6.5. in \cite{Negri2005}, using Lemma \ref{Lemma:PermutationOfTheRuleLKI}. QED
\end{proof}
\begin{lemmaA}[Number of applications of the rule \((\Rightarrow K_{I})\)]\label{Lemma:MinimalApplicationsOfTheRuleRKI}
If a sequent \(S\) is provable in GS-LCK, then there exists the proof of \(S\) such that for each formula \(s:K_{I}A\) in its positive part there are at most \(n(K_{I})\) applications of \((\Rightarrow K_{I})\) iterated on a chain of accessible worlds \(s\overset{I}{\sim}s_{1}, s_{1}\overset{I}{\sim}s_{2}, ...\), with principal formula \(s_{i}:K_{I}A\). The latter proof is called regular.
\end{lemmaA}
\begin{proof}\hspace*{1mm}\\
The Lemma \ref{Lemma:MinimalApplicationsOfTheRuleRKI} is proved by induction on the number \(N\) of series of applications of rule \((\Rightarrow K_{I})\), which make the initial proof non-regular.\\\\
\(<N = 0>\) The proof of the lemma is obtained. \\\\
\(<N > 0>\) \\\\ 
We diminish the inductive paramater in the same way as in the proof of Proposition 6.9. in \cite{Negri2005}. QED
\end{proof}

\begin{theoremA}[Termination of GS-LCK-PROC]\label{theorem:TerminationOfGSLCKPROC}
The procedure GS-LCK-PROC performs terminating proof search for each formula over \((R,\Sigma,\vec O)\).
\end{theoremA}
\begin{proof}\hspace*{1mm}\\
From construction of the procedure GS-LCK-PROC follows that the number of applications of the rules \((K_{I}\Rightarrow)\) and \((\Rightarrow K_{I})\) is finite.\\\\
All the propositional rules reduce the complexity of the root sequent. Since the sets \(N, (R,\Sigma), \vec O\) and the number of applications of the rules \((K_{I}\Rightarrow)\), \((\Rightarrow K_{I})\) are finite, and the requirements are imposed on the rules, the number of applications of the rules \((K_{N}\Rightarrow), (\Rightarrow K_{N}), (OE), (OYR), (CR), (Sub(p) \Rightarrow), (Sub(o^{r}) \Rightarrow), (Ref), (Trans), (Eucl)\) and \((Mon)\) is also finite.\\\\
According to finite number of applications of all rules, the procedure GS-LCK-PROC performs the terminating proof search for any sequent. QED
\end{proof}

\begin{theoremA}[Soundness and completeness of GS-LCK-PROC]\label{theorem:SoundnessAndCompletenessOfGSLCKPROC}
The procedure GS-LCK-PROC is sound and complete over \((R,\Sigma,\vec O)\).
\end{theoremA}
\begin{proof}\hspace*{1mm}\\
From construction of the procedure GS-LCK-PROC follows that if procedure returns "True" for a sequent \(S\), then \(S\) is provable in GS-LCK. If procedure returns "False", then sequent \(S\) is not provable in GS-LCK, according to Lemma \ref{Lemma:MinimalApplicationsOfTheRuleLKI} and Lemma \ref{Lemma:MinimalApplicationsOfTheRuleRKI}. QED
\end{proof}

\begin{theoremA}[Decidibility of LCK]\label{theoremA:1}
Logic LCK is decidable.
\end{theoremA}
\begin{proof}\hspace*{1mm}\\
From Theorem \ref{theorem:SoundnessAndCompletenessOfGSLCKPROC} and Theorem \ref{theorem:TerminationOfGSLCKPROC} follows that GS-LCK-PROC is a decision procedure for logic LCK. QED
\end{proof}

\section*{Conclusions}
Sequent calculus GS-LCK has properties of soundness, completeness, admissibility of cut and structural rules, and invertibility of all rules. Procedure GS-LCK-PROC performs automated terminating proof search for logic of correlated knowledge and also has properties of soundness and completeness. \\\\
Using GS-LCK-PROC, the validity of the formula of any sequent can be determined and inferences can be checked if they follow from some knowledge base. Modelling the knowledge of distributed systems in the logic of correlated knowledge, questions about the systems can be answered automatically. Also soundness, completeness and termination of GS-LCK-PROC show that GS-LCK-PROC is a desicion procedure for logic of correlated knowledge and LCK is decidable logic, which means asking questions about the system we will always get the answer.\\\\
Logic of correlated knowledge expands the range of the applications of family of epistemic logics and captures deeper knowledge of the group of agents in the distributed systems. GS-LCK-PROC allows to reason about correlated knowledge automatically, without human interaction in the reasoning process.\\\\
\bibliographystyle{plain}
\bibliography{Giedra_Alonderis_Automated_proof_search_system_for_logic_of_correlated_knowledge}{}

\begin{thebibliography}{1}

\bibitem{BaltagSmets2010}
A.~Baltag and S.~Smets.
\newblock Correlated knowledge: an epistemic-logic view on quantum
  entanglement.
\newblock {\em International Journal of Theoretical Physics},
  49(12):3005--3021, 2010.

\bibitem{Fagin1992}
R.~Fagin, J.~Y. Halpern, and M.~Y. Vardi.
\newblock What can machines know?: On the properties of knowledge in
  distributed systems.
\newblock {\em Journal of the ACM}, 39(2):328--376, 1992.

\bibitem{Gentzen1934}
G.~Gentzen.
\newblock Untersuchungen uber das logische schliesen. i.
\newblock {\em Mathematische Zeitschrift}, 39(2):176--210, 1934.

\bibitem{Kripke1963}
S.~Kripke.
\newblock Semantical analysis of modal logic i. normal propositional calculi.
\newblock {\em Zeitschrift fur mathematische Logik und Grundlagen der
  Mathematik}, 9:67--96, 1963.

\bibitem{Negri2005}
S.~Negri.
\newblock Proof analysis in modal logic.
\newblock {\em Journal of Philosophical Logic}, 34(5):507--544, 2005.

\bibitem{Negri2001book}
S.~Negri and J.~von Plato.
\newblock {\em Structural Proof Theory}.
\newblock Cambridge University Press, 2001.

\bibitem{Hoek1997}
W.~van~der Hoek and J.-J.~Ch. Meyer.
\newblock A complete epistemic logic for multiple agents--combining distributed
  and common knowledge.
\newblock {\em Epistemic Logic and the Theory of Games and Decisions}, pages
  35--68, 1997.

\end{thebibliography}

\end{document}